\documentclass{llncs}

\usepackage[intlimits]{amsmath}
\usepackage{amsfonts,amssymb}

\usepackage{amsthm}
\usepackage{enumerate}
\usepackage[british]{babel}

\newtheorem{fact}{Fact}

\makeatletter
\newtheorem*{rep@theorem}{\rep@title}
\newcommand{\newreptheorem}[2]{%
\newenvironment{rep#1}[1]{%
 \def\rep@title{#2 \ref{##1}}%
 \begin{rep@theorem}}%
 {\end{rep@theorem}}}
\makeatother
\newreptheorem{theorem}{Theorem}
\newreptheorem{lemma}{Lemma}
\newreptheorem{corollary}{Corollary}

\begin{document}

\title{Sensitivity versus Certificate Complexity of Boolean Functions \thanks{
The research leading to these results has received funding from the European Union Seventh Framework Programme (FP7/2007-2013) under projects RAQUEL (Grant Agreement No. 323970)
and ERC Advanced Grant MQC.}}
\author{Andris Ambainis \and Kri\v{s}j\={a}nis Pr\={u}sis \and Jevg\={e}nijs Vihrovs}
\institute{Faculty of Computing, University of Latvia, Rai\c{n}a bulv. 19, R\=\i ga, LV-1586, Latvia}

\date{}

\maketitle

\begin{abstract}
Sensitivity, block sensitivity and certificate complexity are basic complexity measures of Boolean functions. The famous sensitivity conjecture claims that sensitivity is polynomially related to block sensitivity. However, it has been notoriously hard to obtain even exponential bounds. Since block sensitivity is known to be polynomially related to certificate complexity, an equivalent of proving this conjecture would be showing that the certificate complexity is polynomially related to sensitivity. Previously, it has been shown that $bs(f) \leq C(f) \leq 2^{s(f)-1} s(f) - (s(f)-1)$. In this work, we give a better upper bound of $bs(f) \leq C(f) \leq \max\left(2^{s(f)-1}\left(s(f)-\frac 1 3\right), s(f)\right)$ using a recent theorem limiting the structure of function graphs. We also examine relations between these measures for functions with 1-sensitivity $s_1(f)=2$ and arbitrary 0-sensitivity $s_0(f)$.
\end{abstract}

\section{Introduction}

\emph{Sensitivity} and \emph{block sensitivity} are two well-known 
combinatorial complexity measures of Boolean functions. 
The sensitivity of a Boolean function, $s(f)$, is just the maximum number of
variables $x_i$ in an input assignment $x=(x_1, \ldots, x_n)$ with the property that 
changing $x_i$ changes the value of $f$. 
Block sensitivity, $bs(f)$, is a generalization of sensitivity to the case
when we are allowed to change disjoint blocks of variables.

Sensitivity and block sensitivity are related to the complexity of
computing $f$ in several different computational models, from 
parallel random access machines or PRAMs \cite{Cook_Dwork_Reischuk_1986} 
to decision tree complexity, where block sensitivity has been useful for
showing the complexities of deterministic, probabilistic and quantum decision trees
are all polynomially related \cite{Nisan_1991,Beals+_2001,Buhrman_deWolf_2002}.

A very well-known open problem is the \emph{sensitivity vs. block sensitivity conjecture}
which claims that the two quantities are polynomially related.
This problem is very simple to formulate (so simple that it can be assigned as 
an undergraduate research project). At the same time, the conjecture 
appears quite difficult to solve. It has been known for over 25 years and
the best upper and lower bounds are still very far apart.
We know that block sensitivity can be quadratically larger than sensitivity
\cite{Rubinstein_1995,Virza_2011,Ambainis_Sun_2011} 
but the best upper bounds on block sensitivity in terms of sensitivity are still exponential (of the form
$bs(f)\leq c^{s(f)}$) \cite{Simon_1983,Kenyon_Kutin_2004,A+}.

Block sensitivity is polynomially related to a number of other complexity measures
of Boolean functions: \emph{certificate complexity}, \emph{polynomial degree} and the
number of queries to compute $f$ either deterministically, probabilistically or 
quantumly \cite{Buhrman_deWolf_2002}. 
This gives a number of equivalent formulations for the sensitivity vs. block 
sensitivity conjecture: it is equivalent to asking whether sensitivity is polynomially
related to any one of these complexity measures.

Which of those equivalent forms of the conjecture is the most promising one?
We think that certificate complexity, $C(f)$, is the combinatorially simplest among all 
of these complexity measures. Certificate complexity being at least $c$ 
simply means that there is an input $x=(x_1, \ldots, x_n)$ that is not contained in 
an $(n-(c-1))$-dimensional subcube of the Boolean hypercube on which $f$ is constant.
Therefore, we now focus on the ``sensitivity vs. certificate complexity" form of the conjecture.

\subsection{Prior Work}

The best upper bound on certificate complexity in terms of sensitivity is
\begin{equation} \label{eq:previous}
C_0(f) \leq 2^{s_1(f)-1} s_0(f) - (s_1(f)-1)
\end{equation} due to Ambainis et al. \cite{A+}\footnote{Here,
$C_0$ ($C_1$) and $s_0$ ($s_1$) stand for certificate complexity and sensitivity, restricted to
inputs $x$ with $f(x)=0$ ($f(x)=1$).} The bounds for $C_0(f)$ also hold for $C_1(f)$ symmetrically (in this case, $C_1(f) \leq 2^{s_0(f)-1} s_1(f) - (s_0(f)-1)$), so it is sufficient to focus on the former.

Since $bs(f)  \leq C(f)$, this immediately implies that
\begin{equation}
bs(f) \leq C(f) \leq 2^{s(f)-1} s(f) - (s(f)-1).
\end{equation}

\subsection{Our Results}

In this work, we give improved upper bounds for the ``sensitivity vs. certificate complexity'' problem. Our main technical result is

\begin{theorem} \label{thm:b-cert-sens}
Let $f$ be a Boolean function which is not constant. If $s_1(f) = 1$, then $C_0(f) = s_0(f)$. If $s_1(f) > 1$, then
\begin{equation}
C_0(f) \leq 2^{s_1(f)-1}\left(s_0(f)-\frac 1 3\right).
\end{equation}
\end{theorem}

A similar bound for $C_1(f)$ follows by symmetry. This implies a new upper bound on block sensitivity and certificate complexity in terms of sensitivity:

\begin{corollary} \label{thm:cert-sens}
Let $f$ be a Boolean function. Then
\begin{equation}
bs(f) \leq C(f) \leq \max\left(2^{s(f)-1}\left(s(f)-\frac 1 3\right), s(f)\right).
\end{equation}
\end{corollary}

On the other hand, the function of Ambainis and Sun \cite{Ambainis_Sun_2011} gives the separation of
\begin{equation}
C_0(f) = \left(\frac 2 3 + o(1)\right) s_0(f)s_1(f)
\end{equation}
for arbitrary values of $s_0(f)$ and $s_1(f)$. For $s_1(f) = 2$, we show an example of $f$ that achieves
\begin{equation}
C_0(f) = \left \lfloor \frac{3}{2} s_0(f) \right \rfloor = \left \lfloor \frac{3}{4} s_0(f)s_1(f) \right \rfloor. \label{eq:example}
\end{equation}

We also study the relation between $C_0(f)$ and $s_0(f)$ for functions with low $s_1(f)$, as we think these cases may provide insights into the more general case. 

If $s_1(f)=1$, then $C_0(f)=s_0(f)$ follows from (\ref{eq:previous}). So, the easiest non-trivial case is $s_1(f)=2$, for which (\ref{eq:previous})
becomes $C_0(f) \leq 2 s_0(f) - 1$. 

For $s_1(f) = 2$, we prove a slightly better upper bound of $C_0(f) \leq \frac 9 5 s_0(f)$. We also show that $C_0(f) \leq \frac 3 2 s_0(f)$  for $s_1(f) = 2$ and $s_0(f) \leq 6$  and thus our example (\ref{eq:example}) is optimal in this case. We conjecture that $C_0(f) \leq \frac 3 2 s_0(f)$ is a tight upper bound for $s_1(f) = 2$.


Our results rely on a recent ``gap theorem" by Ambainis and Vihrovs \cite{Ambainis_Vihrovs_2015}
which says that any sensitivity-$s$ induced subgraph $G$ of the Boolean hypercube must be
either of size $2^{n-s}$ or of size at least $\frac{3}{2} 2^{n-s}$ and, 
in the first case, $G$ can only be a subcube obtained by fixing $s$ variables.
Using this theorem allows refining earlier results which 
used Simon's lemma \cite{Simon_1983} -- any sensitivity-$s$ induced subgraph $G$ must be of size at least $2^{n-s}$ --
but did not use any more detailed information about the structure of such $G$.

We think that further research in this direction may uncover more interesting 
facts about the structure of low-sensitivity subsets of the Boolean
hypercube, with implications for the ``sensitivity vs. certificate complexity" conjecture.

\section{Preliminaries}

Let $f: \{0,1\}^n \rightarrow \{0,1\}$ be a Boolean function on $n$ variables. The $i$-th variable of input $x$ is denoted by $x_i$. For an index set $P \subseteq [n]$, let $x^P$ be the input obtained from an input $x$ by flipping every bit $x_i$, $i \in P$.

We briefly define the notions of sensitivity, block sensitivity and certificate complexity.
For more information on them and their relations to other
complexity measures (such as deterministic, probabilistic and quantum decision
tree complexities), we refer the reader to the surveys by Buhrman and de Wolf \cite{Buhrman_deWolf_2002}
and Hatami et al. \cite{Hatami_Kulkarni_Pankratov_2011}.

\begin{definition}
The \emph{sensitivity complexity} $s(f,x)$ of $f$ on an input $x$ is defined as $\left| \left\{ i \,\middle|\, f(x) \neq f\left(x^{\{i\}}\right)\right\} \right|$. The \emph{$b$-sensitivity} $s_b(f)$ of $f$, where $b \in \{0,1 \}$,  is defined as $\max (s(f,x) \mid x \in \{0,1\}^n, f(x)=b)$.  The \emph{sensitivity} $s(f)$ of $f$  is defined as $\max(s_0(f), s_1(f))$.
\end{definition}

\begin{definition}
The \emph{block sensitivity} $bs(f,x)$ of $f$ on input $x$ is defined as the maximum number $t$ such that there are $t$ pairwise disjoint subsets $B_1, \ldots , B_t$ of $[n]$ for which $f(x) \neq f\left(x^{B_i}\right)$. We call each $B_i$ a \emph{block}.  The \emph{$b$-block sensitivity} $bs_b(f)$ of $f$, where $b \in \{0,1 \}$,  is defined as $\max (bs(f,x) \mid x \in \{0,1\}^n, f(x)=b)$.  The \emph{block sensitivity} $bs(f)$ of $f$  is defined as $\max(bs_0(f), bs_1(f))$.
\end{definition}

\begin{definition}
A \emph{certificate} $c$ of $f$ on input $x$ is defined as a partial assignment $c: P \rightarrow \{0,1\}, P \subseteq [n]$ of $x$ such that $f$ is constant on this restriction. We call $|P|$ the \emph{length} of $c$. If $f$ is always 0 on this restriction, the certificate is a \emph{0-certificate}. If $f$ is always 1, the certificate is a \emph{1-certificate}.
\end{definition}

\begin{definition}
The \emph{certificate complexity} $C(f,x)$ of $f$ on input $x$ is defined as the minimum length of a certificate that $x$ satisfies. The \emph{$b$-certificate complexity} $C_b(f)$ of $f$, where $b \in \{0,1 \}$,  is defined as $\max (C(f,x) \mid x \in \{0,1\}^n, f(x)=b)$.  The \emph{certificate complexity} $C(f)$ of $f$  is defined as $\max (C_0(f), C_1(f))$.
\end{definition}

In this work we look at $\{0, 1\}^n$ as a set of vertices for a graph $Q_n$ (called the \emph{$n$-dimensional Boolean cube} or \emph{hypercube}) in which we have an edge $(x, y)$ whenever $x=(x_1, \ldots, x_n)$ and $y=(y_1, \ldots, y_n)$ differ in exactly one position. We look at subsets $S \subseteq \{0,1\}^n$ as subgraphs (induced by the subset of vertices $S$) in this graph.

\begin{definition}
Let $c$ be a partial assignment $c: P \rightarrow \{0,1\}, P \subseteq [n]$. An \emph{$(n-|P|)$-dimensional subcube} of $Q_n$ is a subgraph $G$ induced on a vertex set $\{ x \mid \forall i \in P \, (x_i = c(i)) \}$. It is isomorphic to $Q_{n-|P|}$. We call the value $\dim(G) = n-|P|$ the \emph{dimension} and the value $|P|$ the \emph{co-dimension} of $G$.
\end{definition}

Note that each certificate of length $l$ corresponds to one subcube of $Q_n$ with co-dimension $l$.

\begin{definition}
Let $G$ be a subcube defined by a partial assignment $c : P \rightarrow \{0,1\}, P \subseteq [n]$. Let $c' : P \rightarrow \{0,1\}$ where $c'(i) \neq c(i)$ for exactly one $i \in P$. Then we call the subcube defined by $c'$ a \emph{neighbour subcube} of $G$.
\end{definition}

\begin{definition}
Let $G$ and $H$ be induced subgraphs of $Q_n$. By $G \cap H$ denote the \emph{intersection} of $G$ and $H$ that is the graph induced on $V(G) \cap V(H)$. By $G \cup H$ denote the \emph{union} of $G$ and $H$ that is the graph induced on $V(G) \cup V(H)$. By $G \setminus H$ denote the \emph{complement} of $G$ in $H$ that is the graph induced by $V(G) \setminus V(H)$.
\end{definition}

\begin{definition}
Let $G$ and $H$ be induced subgraphs of $Q_n$. By $R(G, H)$ denote the relative size of $G$ in $H$:
\begin{equation}
R(G, H) = \frac{|G \cap H|}{|H|}.
\end{equation}
\end{definition}

We extend the notion of sensitivity to the induced subgraphs of $Q_n$:

\begin{definition}
Let $G$ be a non-empty induced subgraph of $Q_n$. The \emph{sensitivity} $s(G,Q_n,x)$ of a vertex $x \in Q_n$ is defined as $\left|\left\{ i \,\middle|\, x^{\{i\}} \notin G\right\}\right|$, if $x \in G$, and $\left|\left\{ i \,\middle|\, x^{\{i\}} \in G\right\}\right|$, if $x \notin G$. Then the \emph{sensitivity} of $G$ is defined as $s(G,Q_n) = \max(s(G,Q_n,x) \mid x \in G)$.
\end{definition}

Our results rely on the following generalization of Simon's lemma \cite{Simon_1983}, proved by Ambainis and Vihrovs \cite{Ambainis_Vihrovs_2015}:

\begin{theorem} \label{thm:simon-gen}
Let $G$ be a non-empty induced subgraph of $Q_n$ with sensitivity at most $s$. Then either $R(G, Q_n) = \frac{1}{2^s}$ and $G$ is an $(n-s)$-dimensional subcube or $R(G, Q_n) \geq \frac{3}{2}\cdot\frac{1}{2^s}$.
\end{theorem}

\section{Upper Bound on Certificate Complexity in Terms of Sensitivity}

In this section we prove Corollary \ref{thm:cert-sens}. In fact, we prove a slightly more specific result.

\begin{reptheorem}{thm:b-cert-sens}
Let $f$ be a Boolean function which is not constant. If $s_1(f) = 1$, then $C_0(f) = s_0(f)$. If $s_1(f) > 1$, then
\begin{equation}
C_0(f) \leq 2^{s_1(f)-1}\left(s_0(f)-\frac 1 3\right).
\end{equation}
\end{reptheorem}

Note that a similar bound for $C_1(f)$ follows by symmetry. For the proof, we require the following lemma.

\begin{lemma} \label{thm:lemma}
Let $H_1$, $H_2$, $\ldots$, $H_k$ be distinct subcubes of $Q_n$ such that the Hamming distance between any two of them is at least 2. Take
\begin{equation}
T = \bigcup_{i=1}^k H_i, \quad\quad\quad T' = \left\{ x \,\middle|\, \exists i \,\left(x^{\{i\}} \in T\right) \right\} \setminus T.
\end{equation}
If $T \neq Q_n$, then $|T'| \geq |T|$.
\end{lemma}

\begin{proof}
If $k = 1$, then the co-dimension of $H_1$ is at least 1. Hence $H_1$ has a neighbour cube, so $|T'| \geq |T| = |H_1|$.

Assume $k \geq 2$. Then $n \geq 2$, since there must be at least 2 bit positions for cubes to differ in. We use an induction on $n$.

\textbf{Base case.} $n = 2$. Then we must have that $H_1$ and $H_2$ are two opposite vertices. Then the other two vertices are in $T'$, hence $|T'| = |T| = 2$.

\textbf{Inductive step.} Divide $Q_n$ into two adjacent $(n-1)$-dimensional subcubes $Q_0$ and $Q_1$ by the value of $x_1$. We will prove that the conditions of the lemma hold for each $T \cap Q_b$, $b \in \{0,1\}$. Let $H_u^b = H_u \cap Q_b$. Assume $H_u^b \neq \varnothing$ for some $u \in [k]$. Then either $x_1 = b$ or $x_1$ is not fixed in $H_u$. Thus, if there are two non-empty subcubes $H_u^b$ and $H_v^b$, they differ in the same bit positions as $H_u$ and $H_v$. Thus the Hamming distance between $H_u^b$ and $H_v^b$ is also at least 2. On the other hand, $Q_b \not\subseteq T$, since then $k$ would be at most 1. 

Let $T_b = T \cap Q_b$ and $T_b' = \left\{ x \,\middle|\, x \in Q_b, \exists i \, \left(x^{\{i\}} \in T_b\right) \right\} \setminus T_b$. Then by induction we have that $|T_b'| \geq |T_b|$. On the other hand, $T_0 \cup T_1 = T$ and $T_0' \cup T_1' \subseteq T'$. Thus
\begin{equation}
|T'| \geq |T_0'| + |T_1'| \geq |T_0| + |T_1| = |T|.
\end{equation}

\end{proof}

\begin{proof}[Proof of Theorem \ref{thm:b-cert-sens}]
Let $z$ be a vertex such that $f(z) = 0$ and $C(f, z) = C_0(f)$. Pick a 0-certificate $S_0$ of length $C_0(f)$ and $z \in S_0$. It has $m=C_0(f)$ neighbour subcubes which we denote by $S_1, S_2, \ldots, S_m$.

We work with a graph $G$ induced on a vertex set $\{x \mid f(x) = 1\}$. Since $S_0$ is a minimum certificate for $z$, $S_i \cap G \neq \varnothing$ for $i \in [m]$.

As $S_0$ is a 0-certificate, it gives 1 sensitive bit to each vertex in $G \cap S_i$. Then $s(G \cap S_i, S_i) \leq s_1(f)-1$.

Suppose $s_1(f) = 1$, then for each $i \in [m]$ we must have that $G \cap S_i$ equals to the whole $S_i$. But then each vertex in $S_0$ is sensitive to its neighbour in $G \cap S_i$, so $m \leq s_0(f)$. Hence $C_0(f) = s_0(f)$.

Otherwise $s_1(f) \geq 2$. By Theorem \ref{thm:simon-gen}, either $R(G,S_i) =\frac{1}{2^{s_1(f)-1}}$ or $R(G,S_i) \geq \frac{3}{2^{s_1(f)}}$ for each $i \in [m]$. We call the cube $S_i$ either \emph{light} or \emph{heavy} respectively. We denote the number of light cubes by $l$, then the number of heavy cubes is $m-l$. We can assume that the light cubes are $S_1, \ldots, S_l$.

Let the average sensitivity of $x \in S_0$ be $as(S_0)$. Since each vertex of $G$ in any $S_i$ gives sensitivity 1 to some vertex in $S_0$, $\sum_{i=1}^m R(G, S_i) \leq as(S_0)$. Clearly $as(S_0) \leq s_0(f)$. We have that 

\begin{align}
l \frac{1}{2^{s_1(f)-1}} + (m-l) \frac{3}{2^{s_1(f)}} &\leq as(S_0) \leq s_0(f) \\
m \frac{3}{2^{s_1(f)}} - l \frac{1}{2^{s_1(f)}} &\leq as(S_0) \leq s_0(f). \label{eq:ratio-vs-as}
\end{align}

Then we examine two possible cases.

\bigskip

\textbf{Case 1.} $l \leq (s_0(f)-1)2^{s_1(f)-1}$. Then we have

\begin{align}
m \frac{3}{2^{s_1(f)}} -(s_0(f)-1) \frac{2^{s_1(f)-1}}{2^{s_1(f)}} &\leq as(S_0) \leq s_0(f) \\
m \frac{3}{2^{s_1(f)}} &\leq s_0(f) + \frac 1 2 (s_0(f)-1) \\
m \frac{3}{2^{s_1(f)}} &\leq \frac 3 2 s_0(f) - \frac 1 2 \\
m  &\leq 2^{s_1(f)-1}\left(s_0(f)-\frac 1 3\right).
\end{align}

\bigskip

\textbf{Case 2.} $l = (s_0(f)-1)2^{s_1(f)-1}+\epsilon$ for some $\epsilon > 0$. Since $s_1(f) \geq 2$, the number of light cubes is at least $2(s_0(f)-1) + \epsilon$, which in turn is at least $s_0(f)$.

Let $\mathcal F = \{ F \mid F \subseteq [l], |F| = s_0(f)\}$. Denote its elements by $F_1, F_2, \ldots, F_{|\mathcal F|}$. We examine  $H_1, H_2, \ldots, H_{|\mathcal F|}$ -- subgraphs of $S_0$, where $H_i$ is the set of vertices whose neighbours in $S_j$ are in $G$ for each $j \in F_i$. By Theorem \ref{thm:simon-gen}, $G \cap S_i$ are subcubes for $i \leq l$. Then so are the intersections of their neighbours in $S_0$, including each $H_i$.

Let $N_{i,j}$ be the common neighbour cube of $S_i$ and $S_j$ that is not $S_0$. Suppose $v \in S_0$. Then by $v_i$ denote the neighbour of $v$ in $S_i$. Let $v_{i,j}$ be the common neighbour of $v_i$ and $v_j$ that is in $N_{i,j}$.

We will show that the Hamming distance between any two subcubes $H_i$ and $H_j$, $i \neq j$ is at least 2. 

Assume there is an edge $(u,v)$ such that $u \in H_i$ and $v \in H_j$. Then $u_k \in G$ for each $k \in F_i$. Since $i \neq j$, there is an index $t \in F_j$ such that $t \notin F_i$. The vertex $u$ is sensitive to $S_k$ for each $k \in F_i$ and, since $|F_i| = s_0(f)$, has full sensitivity. Thus $u_t \notin G$. On the other hand, since each $S_k$ is light, $u_k$ has full 1-sensitivity, hence $u_{k,t} \in G$ for all $k \in F_i$. This gives full 0-sensitivity to $u_t$. Hence $v_t \notin G$, a contradiction, since $v \in H_j$ and $t \in F_j$.

Thus there are no such edges and the Hamming distance between $H_i$ and $H_j$ is not equal to 1. That leaves two possibilities: either the Hamming distance between $H_i$ and $H_j$ is at least 2 (in which case we are done), or both $H_i$ and $H_j$ are equal to a single vertex $v$, which is not possible, as then $v$ would have a 0-sensitivity of at least $s_0(f)+1$.

Let $T = \bigcup_{i=1}^{|\mathcal F|} H_i$. We will prove that $T \neq S_0$. Since $s_1(f) \geq 2$, by Theorem \ref{thm:simon-gen} it follows that $\dim(G \cap S_i) = \dim(S_i) - s_1(f)+1 \leq \dim(S_0) - 1$ for each $i \in [l]$. Thus $\dim(H_1) \leq \dim(S_0) - 1$, and $H_1 \neq S_0$. Then it has a neighbour subcube $H_1'$ in $S_0$. But since the Hamming distance between $H_1$ and any other $H_i$ is at least 2, we have that $H_1' \cap H_i = \varnothing$, thus $T$ is not equal to $S_0$.

Therefore, $H_1, H_2, \ldots, H_{|\mathcal F|}$ satisfy all the conditions of Lemma \ref{thm:lemma}. Let $T'$ be the set of vertices in $S_0 \setminus T$ with a neighbour in $T$. Then, by Lemma \ref{thm:lemma}, $|T'| \geq |T|$ or, equivalently, $R(T', S_0) \geq R(T, S_0)$.

Then note that $R(T', S_0) \geq R(T, S_0) \geq \frac{\epsilon}{2^{s_1(f)-1}}$, since $R(G,S_i)= \frac{1}{2^{s_1(f)-1}}$ for all $i \in [l]$, there are a total of $ (s_0(f)-1)2^{s_1(f)-1}+\epsilon$ light cubes and each vertex in $S_0$ can have at most $s_0(f)$ neighbours in $G$.

Let $S_h$ be a heavy cube, and $i \in [|\mathcal F|]$. The neighbours of $H_i$ in $S_h$ must not be in $G$, or the corresponding vertex in $H_i$ would have sensitivity $s_0(f)+1$.

Let $k \in F_i$. As $S_k$ is light, all the vertices in $G \cap S_k$ are fully sensitive, therefore all their neighbours in $N_{k,h}$ are in $G$. Therefore all the neighbours of $H_i$ in $S_h$ already have full 0-sensitivity. Then all their neighbours must also not be in $G$.

This means that vertices in $T'$ can only have neighbours in $G$ in light cubes. But they can have at most $s_0(f)-1$ such neighbours each, otherwise they would be in $T$, not in $T'$. As $R(T', S_0) \geq \frac{\epsilon}{2^{s_1(f)-1}}$, the average sensitivity of vertices in $S_0$ is at most

\begin{align}
as(S_0) &\leq s_0(f) R(S_0 \setminus T', S_0) + (s_0(f)-1)R(T', S_0) \leq \\ &\leq s_0(f) \left(1-\frac{\epsilon}{2^{s_1(f)-1}}\right)+(s_0(f)-1)\frac{\epsilon}{2^{s_1(f)-1}} =\\&= s_0(f)-\frac{\epsilon}{2^{s_1(f)-1}}.
\end{align}

Then by inequality (\ref{eq:ratio-vs-as}) we have
\begin{equation}
m \frac{3}{2^{s_1(f)}} - \left((s_0(f)-1)2^{s_1(f)-1}+\epsilon\right)\frac{1}{2^{s_1(f)}} \leq s_0(f)-\frac{\epsilon}{2^{s_1(f)-1}}.
\end{equation}

Rearranging the terms, we get
\begin{align}
m \frac{3}{2^{s_1(f)}} &\leq \left((s_0(f)-1)2^{s_1(f)-1}+\epsilon\right)\frac{1}{2^{s_1(f)}} + s_0(f)-\frac{\epsilon}{2^{s_1(f)-1}} \\
m \frac{3}{2^{s_1(f)}} &\leq s_0(f) + \frac 1 2(s_0(f)-1) - \frac{\epsilon}{2^{s_1(f)}} \\
m \frac{3}{2^{s_1(f)}} &\leq \frac 3 2 s_0(f) - \frac 1 2 - \frac{\epsilon}{2^{s_1(f)}} \\
m &\leq 2^{s_1(f)-1}\left(s_0(f) - \frac 1 3\right) - \frac{\epsilon}{3}.
\end{align}
\end{proof}

Theorem \ref{thm:b-cert-sens} immediately implies Corollary \ref{thm:cert-sens}:

\begin{proof}[Proof of Corollary \ref{thm:cert-sens}]
If $f$ is constant, then $C(f) = s(f) = 0$ and the statement is true. Otherwise by Theorem \ref{thm:b-cert-sens}

\begin{align}
C(f) &= \max(C_0(f), C_1(f)) \leq \\
& \leq \max_{b \in \{0,1\}}\left(\max\left(2^{s_{1-b}(f)-1}\left(s_b(f)-\frac 1 3\right), s_b(f)\right) \right) \leq \\
& \leq \max\left(2^{s(f)-1}\left(s(f)-\frac 1 3\right), s(f)\right)
\end{align}

On the other hand, $bs(f) \leq C(f)$ is a well-known fact.

\end{proof}


\section{Relation between $C_0(f)$ and $s_0(f)$ for $s_1(f) = 2$}

Ambainis and Sun exhibited a class of functions that achieves the best known separation between sensitivity and block sensitivity, which is quadratic in terms of $s(f)$  \cite{Ambainis_Sun_2011}. This function also produces the best known separation between 0-certificate complexity and 0/1-sensitivity:

\begin{theorem}
For arbitrary $s_0(f)$ and $s_1(f)$, there exists a function $f$ such that
\begin{equation}
C_0(f) = \left(\frac 2 3 + o(1)\right) s_0(f)s_1(f). \label{eq:separation}
\end{equation}
\end{theorem}

Thus it is possible to achieve a quadratic gap between the two measures. As $bs_0(f) \leq C_0(f)$, it would be tempting to conjecture that quadratic separation is the largest possible. Therefore we are interested both in improved upper bounds and in functions that achieve quadratic separation with a larger constant factor.

In this section, we examine how $C_0(f)$ and $s_0(f)$ relate to each other for small $s_1(f)$. If $s_1(f) = 1$, it follows by Theorem \ref{thm:cert-sens} that $C_0(f) = s_0(f)$. Therefore we consider the case $s_1(f) = 2$.

Here we are able to construct a separation that is better than (\ref{eq:separation}) by a constant factor.

\begin{theorem} \label{thm:s1-2-example} There is a function $f$ with $s_1(f) = 2$ and arbitrary $s_0(f)$ such that
\begin{equation} \label{eq:3-2-separation}
C_0(f) = \left\lfloor\frac 3 4 s_0(f)s_1(f)\right\rfloor = \left\lfloor\frac 3 2 s_0(f)\right\rfloor.
\end{equation}
\end{theorem}

\begin{proof}

Consider the function that takes value 1 iff its 4 input bits are in either ascending or descending sorted order. Formally,
\begin{equation}
\textsc{Sort}_4(x) = 1 \Leftrightarrow (x_1 \leq x_2 \leq x_3 \leq x_4) \lor (x_1 \geq x_2 \geq x_3 \geq x_4).
\end{equation}
One easily sees that $C_0(\textsc{Sort}_4) = 3$, $s_0(\textsc{Sort}_4) = 2$ and $s_1(\textsc{Sort}_4) = 2$.

Denote the 2-bit logical AND function by $\textsc{And}_2$. We have $C_0(\textsc{And}_2) = s_0(\textsc{And}_2) = 1$ and $s_1(\textsc{And}_2) = 2$.

To construct the examples for larger $s_0(f)$ values, we use the following fact (it is easy to show, and a similar lemma was proved in \cite{Ambainis_Sun_2011}):
\begin{fact} \label{thm:composition}
Let $f$ and $g$ be Boolean functions. By composing them with OR to $f \lor g$ we get
\begin{align}
C_0(f \lor g) &= C_0(f) + C_0(g), \\
s_0(f \lor g) &= s_0(f) + s_0(g), \\
s_1(f \lor g) &= \max(s_1(f), s_1(g)).
\end{align}
\end{fact}

Suppose we need a function with $k = s_0(f)$. Assume $k$ is even. Then by Fact \ref{thm:composition} for $g = \bigvee_{i=1}^{\frac k 2} \textsc{Sort}_4$ we have $C_0(g) = \frac 3 2 k$. If $k$ is odd, consider the function $g = \left(\bigvee_{i=1}^{\frac{k-1}{2}} \textsc{Sort}_4\right) \lor \textsc{And}_2$. Then by Fact \ref{thm:composition}  we have $C_0(g) = 3 \cdot \frac{k-1}{2} + 1 = \left\lfloor \frac 3 2 k \right\rfloor$.
\end{proof}

A curious fact is that both examples of (\ref{eq:separation}) and Theorem \ref{thm:s1-2-example} are obtained by composing some primitives using OR. The same fact holds for the best examples of separation between $bs(f)$ and $s(f)$ that preceded the \cite{Ambainis_Sun_2011} construction \cite{Rubinstein_1995,Virza_2011}.

We are also able to prove a slightly better upper bound in case $s_1(f) = 2$.

\begin{theorem}
Let $f$ be a Boolean function with $s_1(f) = 2$. Then
\begin{equation}
C_0(f) \leq \frac 9 5 s_0(f).
\end{equation}
\end{theorem}

\begin{proof}
Let $z$ be a vertex such that $f(z) = 0$ and $C(f,z) = C_0(f)$. Pick a 0-certificate $S_0$ of length $m=C_0(f)$ and $z \in S_0$. It has $m$ neighbour subcubes which we denote by $S_1$, $S_2$, $\ldots$, $S_{m}$. Let $n'=n-m = \dim(S_i)$ for each $S_i$.

We work with a graph $G$ induced on a vertex set $\{x \mid f(x) = 1\}$. Let $G_i = G \cap S_i$. As $S_0$ is a minimal certificate for $z$, we have $G_i \neq \varnothing$ for each $i \in [m]$. Since any $v \in G_i$ is sensitive to $S_0$, we have $s(G_i, S_i) \leq 1$. Thus by Theorem \ref{thm:simon-gen} either $G_i$ is an $(n'-1)$-subcube of $S_i$ with $R(G_i : S_i) = \frac{1}{2}$ or $R(G_i : S_i) \geq \frac{3}{4}$. We call $S_i$ \emph{light} or \emph{heavy}, respectively.

Let $N_{i,j}$ be the common neighbour cube of $S_i$, $S_j$ that is not $S_0$. Let $G_{i,j} = G \cap N_{i,j}$. Suppose $v \in S_0$. Let $v_i$ be the neighbour of $v$ in $S_i$. Let $v_{i,j}$ be the neighbour of $v_i$ and $v_j$ in $N_{i,j}$.

Let $S_i, S_j$ be light. By $G_i^0, G_j^0$ denote the neighbour cubes of $G_i, G_j$ in $S_0$. We call $\{S_i, S_j\}$ a \emph{pair}, iff $G_i^0 \cup G_j^0 = S_0$. In other words, a pair is defined by a single dimension. Also we have either $z_i \notin G$ or $z_j \notin G$: we call the corresponding cube the \emph{representative} of this pair.

\begin{proposition}\label{thm:pairs}
Let $\mathcal P$ be a set of mutually disjoint pairs of the neighbour cubes of $S_0$. Then there exists a 0-certificate $S_0'$ such that $z \in S_0'$, $\dim(S_0') = \dim(S_0)$ and $S_0'$ has at least $|\mathcal P|$ heavy neighbour cubes.
\end{proposition}

\begin{proof}
Let $\mathcal R$ be a set of mutually disjoint pairs of the neighbour cubes of $S_0$. W.l.o.g. let $S_1, \ldots, S_{|\mathcal R|}$ be the representatives of $\mathcal R$. Let $F_i$ be the neighbour cube of $S_i \setminus G$ in $S_0$. Let $B_{\mathcal R} = \bigcap_{i=1}^{|\mathcal R|} F_i$. Suppose $S_0+x$ is a coset of $S_0$ and $x_t = 0$ if the $t$-th dimension is not fixed in $S_0$: let $B_{\mathcal R}(S_0+x)$ be $B_{\mathcal R}+x$.

Pick $\mathcal R \subseteq \mathcal P$ with the largest size, such that for each two representatives $S_i$, $S_j$ of $\mathcal R$, $B_{\mathcal R}(N_{i,j})$ is a 0-certificate.

We will prove that the subcube $S_0'$ spanned by $B_{\mathcal R}, B_{\mathcal R}(S_1), \ldots, B_{\mathcal R}\left(S_{|\mathcal R|}\right)$ is a 0-certificate. It corresponds to an $|\mathcal R|$-dimensional hypercube $Q_{|\mathcal R|}$ where $B_{\mathcal R}(S_0+x)$ corresponds to a single vertex for each coset $S_0+x$ of $S_0$.

Let $T \subseteq Q_{|\mathcal R|}$ be the graph induced on $\{v \mid B_{\mathcal R}(H) \text{ corresponds to } v, B_{\mathcal R}(H)\allowbreak \text{is not a 0-certificate}\}$. Then we have $s(T, Q_{|\mathcal R|}) \leq 2$. Suppose $B_{\mathcal R}$ corresponds to $0^{|\mathcal R|}$. Let $L_d$ be the set of $Q_{|\mathcal R|}$ vertices that are at distance $d$ from $0^{|\mathcal R|}$. We prove by induction that $L_d \cap T = \varnothing$ for each $d$.

\begin{proof}
\textbf{Base case.} $d \leq 2$. The required holds since all $B_{\mathcal R}, B_{\mathcal R}(S_i), B_{\mathcal R}(N_{i, j})$ are 0-certificates.

\textbf{Inductive step.} $d \geq 3$. Examine $v \in L_d$. As $v$ has $d$ neighbours in $L_{d-1}$, $L_{d-1} \cap T = \varnothing$ and $s(T, Q_{|\mathcal R|}) \leq 2$, we have that $v \notin T$.
\end{proof}

Let $k$ be the number of distinct dimensions that define the pairs of $\mathcal R$, then $k \leq |\mathcal R|$. Hence $\dim(S_0') = |\mathcal R| + \dim(B_{\mathcal R}) = |\mathcal R| + (\dim(S_0)-k) \geq \dim(S_0)$. But $S_0$ is a minimal 0-certificate for $z$, therefore $\dim(S_0') = \dim(S_0)$.

Note that a light neighbour $S_i$ of $S_0$ is separated into a 0-certificate and a 1-certificate by a single dimension, hence we have $s(G,S_i,v) = 1$ for every $v \in S_i$. As $S_i$ neighbours $S_0$, every vertex in its 1-certificate is fully sensitive. The same holds for any light neighbour $S_i'$ of $S_0'$.

Now we will prove that each pair in $\mathcal P$ provides a heavy neighbour for $S_0'$. Let $\{S_a, S_b\} \in \mathcal P$, where $S_a$ is the representative. We distinguish two cases:
\begin{itemize}
\item $B_{\mathcal R}(S_b)$ is a 1-certificate. Since $S_b$ is light, it has full 1-sensitivity. Therefore, $v \in G$ for all $v\in B_R(N_{i, b})$, for each $i \in [|\mathcal R|]$. Let $S_b'$ be the neighbour of $S_0'$ that contains $B_{\mathcal R}(S_b)$ as a subcube. Then for each $v \in B_{\mathcal R}(S_b)$ we have $s(G,S_b',v) = 0$. Hence $S_b'$ is heavy.
\item Otherwise, $\{S_a, S_b\}$ is defined by a different dimension than any of the pairs in $\mathcal R$. Let $\mathcal R' = \mathcal R \cup \{S_a, S_b\}$. Examine the subcube $B_{\mathcal R'}$. By definition of $\mathcal R$, there is a representative $S_i$ of $\mathcal R$ such that $B_{\mathcal R'}(N_{i,a})$ is not a 0-certificate. Let $S'_a$ be the neighbour of $S_0'$ that contains $B_{\mathcal R}(S_a)$ as a subcube. Then there is a vertex $v \in B_{\mathcal R'}(S_a)$ such that $s(G, S_a', v) \geq 2$. Hence $S_a'$ is heavy.
\end{itemize}

\end{proof}

Let $\mathcal P$ be the largest such set. Let $l$ and $h = m- l$ be the number of light and heavy neighbours of $S_0$, respectively. Each pair in $\mathcal P$ gives one neighbour in $G$ to each vertex in $S_0$. Now examine the remaining $l-2 | \mathcal P|$ light cubes. As they are not in $\mathcal P$, no two of them form a pair. Hence there is a vertex $v \in S_0$ that is sensitive to each of them. Then $s_0(f) \geq s_0(f,v) \geq |\mathcal P| + (l-2 | \mathcal P|) = l - |\mathcal P|$. Therefore $|\mathcal P| \geq l - s_0(f)$. 

Let $q$ be such that $m=q s_0(f)$. Then there are $q s_0(f)-l$ heavy neighbours of $S_0$. On the other hand, by Proposition \ref{thm:pairs}, there exists a minimal certificate $S_0'$ of $z$ with at least $l-s_0(f)$ heavy neighbours. Then $z$ has a minimal certificate with at least $\frac{(q s_0(f)-l)+(l-s_0(f))}{2} = \frac{q-1}{2}\cdot s_0(f)$ heavy neighbour cubes.

W.l.o.g. let $S_0$ be this certificate. Then $l = qs_0(f) - h \leq (q - \frac{q-1}{2}) s_0(f) = \frac{q+1}{2}\cdot s_0(f)$. As each $v \in G_i$ for $i \in [m]$ gives sensitivity 1 to its neighbour in $S_0$, 
\begin{align}
l \frac 1 2 + h\frac 3 4 &\leq s_0(f).
\end{align}

Since the constant factor at $l$ is less than at $h$, we have
\begin{equation}
\frac{q+1}{2} \cdot s_0(f) \cdot \frac 1 2 + \frac{q-1}{2} \cdot s_0(f)  \cdot \frac 3 4 \leq s_0(f)
\end{equation}

By dividing both sides by $s_0(f)$ and simplifying terms, we get $q \leq \frac 9 5$.

\end{proof}

This result shows that the bound of Theorem \ref{thm:cert-sens} can be improved. However, it is still not tight. For some special cases, through extensive casework we can also prove the following results:

\begin{theorem} \label{thm:s0-3}
Let $f$ be a Boolean function with $s_1(f) = 2$ and $s_0(f) \geq 3$. Then
\begin{equation}
C_0(f) \leq 2s_0(f) - 2.
\end{equation}
\end{theorem}

\begin{theorem} \label{thm:s0-5}
Let $f$ be a Boolean function with $s_1(f) = 2$ and $s_0(f) \geq 5$. Then
\begin{equation}
C_0(f) \leq 2s_0(f) - 3.
\end{equation}
\end{theorem}

These theorems imply that for $s_1(f) = 2$, $s_0(f) \leq 6$ we have $C_0(f) \leq \frac 3 2 s_0(f)$, which is the same separation as achieved by the example of Theorem \ref{thm:s1-2-example}. This leads us to the following conjecture:

\begin{conjecture} \label{thm:s1-2-conj}
Let $f$ be a Boolean function with $s_1(f) = 2$. Then
\begin{equation}
C_0(f) \leq \frac 3 2 s_0(f).
\end{equation}
\end{conjecture}

We consider $s_1(f) = 2$ to be the simplest case where we don't know the actual tight upper bound on $C_0(f)$ in terms of $s_0(f), s_1(f)$. Proving Conjecture \ref{thm:s1-2-conj} may provide insights into relations between $C(f)$ and $s(f)$ for the general case.

\bibliographystyle{abbrv}
\bibliography{bibliography}

\appendix
\section{Proof of Theorem \ref{thm:s0-3}}

\begin{proof}
Let $G$ be a graph induced on a vertex set $\{x \mid f(x) = 1\}$. Suppose $z$ is a vertex such that $f(z) = 0$ and $C(f, z) = C_0(f)$. Pick a 0-certificate $S_0$ of length $C_0(f)$ and $z \in S_0$. It has $m=C_0(f)$ neighbour subcubes which we denote by $S_1, S_2, \ldots, S_m$. Since $S_0$ is a minimum certificate for $z$, we have that $S_i \cap G \neq \varnothing$ for $i \in [m]$. Let the dimension of $S_0$ be $n'$.

Let $N_{i,j}$ be the common neighbour cube of $S_i$ and $S_j$ that is not $S_0$. Suppose $v \in S_0$. Then by $v_i$ denote the neighbour of $v$ in $S_i$. Let $v_{i,j}$ be the common neighbour of $v_i$ and $v_j$ that is in $N_{i,j}$.

As $S_0$ is a 0-certificate, each vertex of $G$ in any $S_i$ is sensitive to its neighbour in $S_0$. Then $s(G \cap S_i, S_i) \leq 1$. By Theorem \ref{thm:simon-gen}, it follows that either $R(G, S_i) = \frac 1 2$ and $G \cap S_i$ is a $(n'-1)$-dimensional subcube or $R(G, S_i) \geq \frac 3 4$. We call such subcubes \emph{light} or \emph{heavy} respectively. Let the number of light cubes be $l$, then the number of heavy cubes is $m-l$. Assume the light cubes are $S_1, S_2, \ldots, S_l$.

Each vertex of $G$ in any $S_i$ gives sensitivity 1 to its neighbour in $S_0$, thus
\begin{equation}
\sum_{i=1}^m R(G, S_i) \leq s_0(f).
\end{equation}
Hence
\begin{align}
l \frac 1 2 + (m-l)\frac 3 4 &\leq s_0(f), \\
3m - 4s_0(f) & \leq l.
\end{align}

Assume on the contrary that $C_0(f) > 2s_0(f)-2$ or equivalently $m \geq 2s_0(f) - 1$. In that case $l \geq 3(2s_0(f)-1) - 4s_0(f) = 2s_0(f) - 3$.

Let $s_0(f) = 3$, then $l \geq 3$. First assume $l = 3$. Then $S_1, S_2, S_3$ are light and $S_4, S_5$ are heavy. This implies
\begin{equation}
3 = s_0(f) \geq \sum_{i=1}^m R(G, S_i) \geq 3\cdot \frac 1 2 + 2 \cdot \frac 3 4 = 3.
\end{equation}
Thus for each $j \in [4;5]$ we have $R(G, S_j) = \frac 3 4$, which means that $R(S_j \setminus G, S_j) = \frac 1 4$. By Theorem \ref{thm:simon-gen}, $s(S_j \setminus G, S_j) \geq 2$. Let $v_j \notin G$ be a vertex with $s(S_j \setminus G, S_j, v_j) = 2$. Assume $v_i \in G$ for some $i \in [3]$. As $v_i$ has full sensitivity, its neighbour $v_{i,j} \in G$. Hence there is at most one such $i \in [3]$ that $v_i \in G$. But then $v \in S_0$ can only have neighbours in $G$ in one light and one heavy cube. Hence $v$ does not have full sensitivity: a contradiction, since $\sum_{i=1}^m R(G, S_i) = 3$. Thus we have $l \geq 4$ for $s_0(f) = 3$.

Now let $s_0(f) \geq 3$. Assume there is a vertex $v \in S_0$ that has $s_0(f)$ neighbours in $G$ among light cubes. Let $S_i$ be a light cube with $v_i \in G$. If $s_0(f) \geq 4$, then $l \geq 2s_0(f)-3 > s_0(f)$; if $s_0(f) = 3$, then $l \geq 4 > 3$. Thus $l > s_0(f)$. Hence there is a light cube $S_j$ such that $v_j \notin G$. Since $v_i$ has full sensitivity, $v_{i,j} \in G$, and there are $s_0(f)$ such $i$. But $v_j$ is also sensitive to one neighbour in $S_j$; hence $v_j$ has sensitivity $s_0(f)+1$, a contradiction.

Thus any vertex $v \in S_0$ has at most $s_0(f)-1$ neighbours in $G$ in light cubes. Then $l \leq 2(s_0(f)-1)$, otherwise we would have a contradiction by the pigeonhole principle. If there are no heavy cubes, then $m = l \leq 2s_0(f) - 2$ and we are done. Otherwise there is a heavy cube $S_h$. Let $T$ be the subset of vertices in $S_0$ that each has exactly $s_0(f)-1$ neighbours in $G$ in light cubes. Since $l \geq 2s_0(f) - 3$, we have $R(T, S_0) \geq \frac 1 2$.

Pick a vertex $v \in S'$. Let $S_i, S_j$ be light cubes with $v_i \in G$ and $v_j \notin G$. If $s_0(f) \geq 4$, then by $l \geq 2s_0(f) - 3$ we have that the number of choices for $j$ is at least $(2s_0(f)-3) - (s_0(f)-1) = 2$; if $s_0(f) = 3$, then since $l \geq 4$, this number is also at least 2. Since $v_i$ has full sensitivity, $v_{i,j} \in G$, and there are $s_0(f)-1$ choices for $i$. On the other hand, as $S_j$ is a light cube, $v_j$ is sensitive to a neighbour in $S_j$. Hence $v_j$ has full sensitivity, so its neighbour $v_{j,h} \notin G$. But then $v_h$ has at least 3 neighbours not in $G$, and, as $s_1(f) = 2$, we have $v_h \notin G$.

This shows that for a vertex $v \in S'$, its neighbour in $S_h$ does not belong to $G$. Let $S_h'$ be the set of $S'$ neighbours in $S_h$. Then $R(S_h \setminus G, S_h) \geq R(S_h', S_h) = R(S', S_0) \geq \frac 1 2$: a contradiction, since $S_h$ is a heavy cube.
\end{proof}

\section{Proof of Theorem \ref{thm:s0-5}}

\begin{lemma} \label{thm:fully-sensitive}
 Let $G$ be a non-empty subgraph of $Q_n$ induced on the vertex set $\{x \mid f(x) = 1\}$ of a function $f$ with $s_1(f) \leq 1$. Then either $G = Q_n$, or $G$ is an $(n-1)$-dimensional subcube, or the number of fully sensitive vertices in $G$ is at least $2 |Q_n \setminus G|$. Furthermore, in the last case each vertex in $Q_n \setminus G$ has a sensitivity of at least 2.
\end{lemma}
\begin{proof}
We examine the induced graph $Q_n \setminus G$. Each connected component in this graph must be a subcube, otherwise some of the vertices of $G$ in the smallest subcube containing the component would have sensitivity at least 2. Furthermore, the Hamming distance between any two of these subcubes is at least 3, otherwise the vertices between them would have sensitivity at least 2.

If there are no such subcubes, $G=Q_n$. If there is such an $(n-1)$-subcube, it must be the only one and $G$ is its opposite $(n-1)$-subcube. Otherwise each of these subcubes is an $(n-2)$-subcube or smaller. Therefore each vertex in them has at least 2 neighbours in $G$. Since $s_1(f)\leq1$, these 2 neighbours are fully sensitive and must be different for each such vertex, thus there are at least  $2 |Q_n \setminus G|$ of them.
\end{proof}

We denote a subcube that can be obtained by fixing some continuous sequence $b$ of starting bits by $Q_b$. For example, $Q_0$ and $Q_1$ can be obtained by fixing the first bit and $Q_{01}$ can be obtained by fixing the first two bits to 01. We use a wildcard * symbol to indicate that the bit in the corresponding position is not fixed. For example, by $Q_{*10}$ we denote a cube obtained by fixing the second and the third bit to 10.

\begin{proof}[Proof of Theorem \ref{thm:s0-5}]

Assume on the contrary that such a function exists. Denote $m=C_0(f)= 2 s_0(f)-2$ and $k=s_0(f) \geq 5$.
 
We work with a graph $G$ induced on a vertex set $\{x \mid f(x) = 1\}$. W.l.o.g. let $z$ be a vertex such that $f(z) = 0$ and $C(f,z) = C_0(f)$. Pick a 0-certificate $S_0$ of length $m=C_0(f)$ and $z \in S_0$. It has $m$ neighbour subcubes which we denote by $S_1$, $S_2$, $\ldots$, $S_{m}$. Let $n'=n-m$, the number of dimensions in each of $S_i$.

Each vertex of $G$ in any $S_i$ gives sensitivity 1 to some vertex in $S_0$, therefore $\sum_{i=1}^m R(G, S_i) \leq s_0(f) =k$.

Let $G_i = G \cap S_i$. Each $G_i$ is nonempty, otherwise we would obtain a shorter certificate for $z$.

Since any 1-vertex in each of $S_i$ is sensitive to $S_0$, we have $s(G, S_i) \leq s_1(f)-1 = 1$. Thus by Theorem \ref{thm:simon-gen} $G_i$ can be either an $(n'-1)$-subcube of $S_i$ with $R(G_i, S_i) = \frac{1}{2}$ or $R(G_i, S_i) \geq \frac{3}{4}$. We will call these cubes \emph{light} or \emph{heavy}, respectively.

Let $F_i$ be the set of fully sensitive vertices in $G_i$.  Note that $F_i = G_i$ for light cubes and $|F_i|= 2 |S_i \setminus G_i|$ for heavy cubes by Lemma \ref{thm:fully-sensitive}.

By $N_{i,j}$ denote the common neighbour cube of $S_i$, $S_j$ that is not $S_0$. For a vertex $v \in S_0$, denote by $v_i$ its neighbour in $S_i$. For $v_i, v_j, i \neq j$, by $v_{i,j}$ denote their common neighbour in $N_{i,j}$.

\bigskip

We will first show that no vertex in $S_0$ has $k$ fully sensitive neighbours. Assume that there exists such a vertex $v$. W.l.o.g. assume that its $k$ fully sensitive neighbours are in $S_1, \ldots, S_k$. Then examine $v_m$. As $v$ is already fully sensitve, $v_m \notin G$. As $v_i \in G$ and fully sensitive for each $i \in [k]$, we have that each $v_{i,m}$ is also in $G$.

But then it follows by induction that no vertex in $S_m$ can be in $G$. As a basis, we have that $v_m$ is also fully sensitive, therefore all of its neighbours in $S_m$ must also not be in $G$.  Now assume that all vertices in $S_m$ with a distance to $v_m$ no more that $i \geq 1$ are not in $G$, then examine a vertex $u_m$ at distance $i+1$. $u_m$ differs from $v_m$ in at least $i+1 \geq 2$ bits, therefore it has at least 2 neighbours closer to $v_m$ which are in $S_m \setminus G$. But we have that $G_m$ has sensitivity at most $s_1(f)-1 = 1$ inside $S_m$, therefore $u_m \notin G$.

But then $G_m = \varnothing$, a contradiction.

\bigskip

Since $\sum_{i=1}^m R(G, S_i) \leq k$, we have that there are at most 4 heavy cubes. We will now examine each possible number of heavy cubes separately.

\bigskip

First we examine the case where there are $m-4$ light cubes and 4 heavy cubes. Since $\sum_{i=1}^m R(G, S_i) \leq k$, for each heavy cube $S_i$ it holds that $R(G_i : S_i) = \frac{3}{4}$ exactly. 

Then by Lemma \ref{thm:fully-sensitive}, $R(F_i, S_i) \geq \frac{1}{2}$ for all $i$. Then  $\sum_{i=1}^m R(F_i, S_i) \geq k-1$. Since no vertex in $S_0$ can have $k$ fully sensitive neighbours, we have that each vertex in $S_0$ has exactly $k-1$ fully sensitive neighbours.

Now examine a vertex $v \in S_0$ which has at least $k-3$ neighbours in $G$ in light cubes, such a vertex must exist as there are a total of $2 k -6$ light cubes in this case. W.l.o.g. assume that its $k-1$ fully sensitive neighbours are in $S_1, \ldots, S_{k-1}$. As $v$ already has $k-3$ neighbours in $G$ in light cubes and can have at most $k$ neighbours in $G$, it must have a neighbour not in $G$ in a heavy cube, W.l.o.g. let that heavy cube be $S_m$. Then examine $v_m \notin G$. As $v_i \in G$ and fully sensitive for each $i \in [k-1]$, we have that each $v_{i,m}$ is also in $G$. In addition, $v_m$ has at least 2 neighbours in $G$ in $S_m$ by Lemma \ref{thm:fully-sensitive}. But then it has sensitivity at least $k+1$, a contradiction.

\bigskip

Now we examine the case where there are no heavy cubes. $R(F_i, S_i) = \frac{1}{2}$ for all $i$ and $\sum_{i=1}^m R(F_i, S_i) = k-1$. Since no vertex in $S_0$ can have $k$ fully sensitive neighbours, we have that each vertex in $S_0$ has exactly $k-1$ fully sensitive neighbours in $G$.

If no two of $G_i$ were opposite subcubes in $Q_{n'}$, they would all overlap  in $Q_{n'}$, giving a vertex in $S_0$ with sensitivity $m > k$. Therefore at least 2 of them are opposite, W.l.o.g. let $G_1= Q_0 \cap S_1$ and $G_2=Q_1 \cap S_2$. 

Now W.l.o.g. examine the case where $z \in Q_0$. Then we examine the 0-certificate $S_0'= Q_0 \cap (S_0 \cup S_2)$. It is also a minimal certificate, but has a heavy neighbour $Q_0 \cap (S_1 \cup N_{1,2})$, therefore we have reduced this case to  the cases where at least one of $S_i$ is heavy.

\bigskip

We now examine the case where there are $m-1$ light cubes and 1 heavy cube, let them be $S_1, \ldots, S_{m-1}$ and $S_m$ respectively. Note that $m-1 \geq k+1$, as $k \geq 5$. Then $\sum_{i=1}^{m-1} R(F_i, S_i) = \frac {1}{2} (2k-3)=k- \frac{3}{2}$. Since no vertex in $S_0$ can have $k$ fully sensitive neighbours, we have that half of all vertices in $S_0$ have exactly $k-1$ fully sensitive neighbours in $G$  in light cubes.

We now examine one such vertex $v \in S_0$. W.l.o.g. assume that its $k-1$ fully sensitive neighbours are in $S_1, \ldots, S_{k-1}$. Then for all $i \in [k-1]$ we have that $v_{i,k}, v_{i,k+1}, v_{i,m}$ are all in $G$.  Since $S_k$ and $S_{k+1}$ are light, we have that $v_k$ and $v_{k+1}$ are fully sensitive. Therefore $v_{k,m}$ and $v_{k+1,m}$ are not in $G$. Therefore $v_m$ has at least 3 neighbours not in $G$ and is also not in G. But then $v_m$ has at least 2 additional neighbours in $G$ by Lemma \ref{thm:fully-sensitive}. Therefore it has sensitivity at least $k+1$, a contradiction.

\bigskip

We now examine the case where there are $m-2$ light cubes and 2 heavy cubes, let them be $S_1, \ldots, S_{m-2}$ and $S_{m-1}, S_{m}$ respectively.  Then $\sum_{i=1}^{m-2} R(F_i, S_i) = k-2$.  If there exists a vertex $v \in S_0$ with $k-1$ neighbours in $G$ in light cubes, we can derive a contradiction similarly to the case with 1 heavy cube (since $m-2 \geq k+1$ as well). Therefore no such vertex exists and each vertex in $S_0$ has exactly $k-2$ neighbours in $G$ in light cubes.

We will next show that $G_{m-1}=S_{m-1}$ and $G_m=S_m$. Assume otherwise, that there exists a vertex in $S_{m-1}, S_m$ not in $G$.  W.l.o.g. let $v_m \notin G$. By Lemma \ref{thm:fully-sensitive}, it has a fully sensitive neighbour $u_m \in G_m$. 
We have that exactly $k-2$ of $u_1, \ldots, u_{m-2}$ are in G, W.l.o.g. let them be $u_1, \ldots , u_{k-2}$. Now examine a vertex $u_i \notin G, i \in [k-1,m-2]$. There are $k-2$ such $u_i$. 
We have that $u_{j,i}$ belongs to $G$ for all $j \in [k-2]$. Additionally so does $u_{i,m}$. Finally, $u_i$ has sensitivity 1 in $S_i$. Then it is fully sensitive, therefore $u_{i,m-1} \notin G$. 
Then $u_{m-1}$ has at least $k-1$ neighbours not in $G$ and it is also not in $G$. But  $u_{j,m-1}$ belongs to $G$ for all $j \in [k-2]$ and so does $u_{m-1,m}$. By Lemma \ref{thm:fully-sensitive}, $u_{m-1}$ also has sensitivity 2 in $S_{m-1}$, giving it a total sensitivity of $k+1$, a contradiction.

If no two light cubes were opposite subcubes in $Q_{n'}$, they would all overlap  in $Q_{n'}$, giving a vertex in $S_0$ with $m-2$ neighbours in $G$ in light cubes. Therefore at least 2 of them are opposite, W.l.o.g. let them be $G_1$ and $G_2$. Then we have that each vertex in $S_0$ has exactly $k-3$ neighbours in $G$ in $G_3, \ldots, G_{m-2}$ and again we have that at least 2 of them must be opposite, W.l.o.g. let them be $G_3, G_4$. Similarly we obtain that $G_{2i-1}, G_{2i}$ are opposite for all $i\in[\frac{m-2}{2}]$. Now we examine 2 cases:
\begin{enumerate}
\item All pairs of light cubes consist of the same subcubes. W.l.o.g. we have that $G_{2i-1} = Q_0$ and $G_{2i} = Q_1$ for all $i\in[\frac{m-2}{2}]$. Now w.l.o.g. examine the case where $z \in Q_0$. Then we examine the 0-certificate $S_0'= Q_0 \cap (S_0 \cup S_2)$. It is also a minimal certificate, but has heavy neighbours $Q_0 \cap (S_{2i-1} \cup N_{2i-1,2})$ for all $i\in[\frac{m-2}{2}]$. Therefore we have reduced this case to the cases where at least 3 of $S_i$ are heavy.

\item At least 2 of the pairs of light cubes consist of different subcubes. W.l.o.g. we have that $G_1=Q_0, G_2=Q_1, G_3=Q_{*0}, G_4=Q_{*1}$. W.l.o.g. assume that $z \in Q_{00}$. 

Then we examine the 0-certificate $S_0'= Q_0 \cap (S_0 \cup S_2)$. It is also a minimal certificate, and has a heavy neighbour $Q_0 \cap (S_1 \cup N_{1,2})$. If it has 2 more heavy neighbours we have reduced this case to the cases where at least 3 of $S_i$ are heavy. Otherwise at least one of $Q_0 \cap (S_{m-1} \cup N_{2,m-1})$, $Q_0 \cap (S_m \cup N_{2,m})$ is light. W.l.o.g. assume it is $Q_0 \cap (S_{m-1} \cup N_{2,m-1})$. Since $G_{m-1}=S_{m-1}$, this means that $Q_0 \cap N_{2,m-1}$ does not have any vertices in $G$. Then $Q_0 \cap S_{m-1}$ is fully sensitive.

Now examine the 0-certificate $S_0'' =  Q_{*0} \cap (S_0 \cup S_4)$. It is also a minimal certificate, and has a heavy neighbour $Q_{*0} \cap (S_3 \cup N_{3,4})$. If it has 2 more heavy neighbours we have reduced this case to the cases where at least 3 of $S_i$ are heavy. Otherwise at least one of $Q_{*0} \cap (S_{m-1} \cup N_{4,m-1})$, $Q_{*0} \cap (S_m \cup N_{4,m})$ is light. As $Q_0 \cap S_{m-1}$ is fully sensitive,  we have that $Q_0 \cap N_{4,m-1} \subset G$ and it must be the case that $Q_{*0} \cap (S_m \cup N_{4,m})$ is light. Since $G_m=S_m$, this means that $Q_{*0} \cap N_{4,m}$ does not have any vertices in $G$. Then $Q_{*0} \cap S_m$ is fully sensitive.

Now note that, as $G_5$ and $G_6$ are opposite subcubes, one of $z_5$, $z_6$ must not belong to $G$.
W.l.o.g. let it be $z_5$. As each vertex in $S_0$ has exactly $k-2$ neighbours in $G$ in light cubes, $z_5$ has at least $k-2$ neighbours in $G$ in $N_{5,j}$ for $j \in [m-2]$. 
Since $Q_0 \cap S_{m-1}$ and $Q_{*0} \cap S_m$ are fully sensitive, $z_{5,m-1}$ and $z_{5,m}$ are also in $G$. 
As $z_5$ also has a neighbour in $G$ in $S_5$, it has a sensitivity of at least $k+1$, a contradiction.
\end{enumerate}

\bigskip

We are left with the case where there are $m-3$ light cubes and 3 heavy cubes, let them be $S_1, \ldots, S_{m-3}$ and $S_{m-2}, S_{m-1}, S_m$ respectively.

We will first show that no vertex in $S_0$ has $k-1$ neighbours in $G$  in light cubes. Assume on the contrary that such a vertex $v \in S_0$ exists. W.l.o.g. assume that its $k-1$ fully sensitive neighbours are in $S_1, \ldots, S_{k-1}$. Then for all $i \in [k-1], j \in [m-2,m]$ we have that $v_{i,j} \in G$.  Then, if $v_j \notin G$, it would have sensitivity $k+1$ ($k-1$ in $v_{i,j}$, 2 in $S_j$ by Lemma \ref{thm:fully-sensitive}). But if $v_j \in G$ for all $j \in [m-2,m]$, $v$ would have sensitivity $k+2$. Therefore no such $v$ exists.

Now note that $\sum_{i=1}^{m-3} R(F_i, S_i) =k - \frac{5}{2}$. Therefore at least half the vertices of $S_0$ have $k-2$ neighbours in $G$ in light cubes. Examine one such vertex $v$. W.l.o.g. assume that $v_1,\ldots, v_{k-2}$ are in $G$. We will now show that at most one of $v_{m-2}, v_{m-1}, v_m$ is in $G$. Assume that on the contrary 2 of them are (if all 3 were, the sensitivity of $v$ would be $k+1$). W.l.o.g. let $v_{m-2} \notin G$, $v_{m-1}, v_m \in G$. As $v_{i,m-2} \in G$ for $i \in [k-2]$ and $v_{m-2}$ has sensitivity 2 in $S_{m-2}$ by Lemma \ref{thm:fully-sensitive}, we have that $v_{m-2}$ is fully sensitive. Therefore $v_{m-2,m-1}$ and $v_{m-2,m}$ are both not in $G$. Then $v_{m-1}$ and $v_m$ are both fully sensitive and $v_{k-1,m-1}, v_{k-1,m} \in G$. But  $v_{i,k-1} \in G$ for $i \in [k-2]$  and $v_{k-1}$ has sensitivity 1 in $S_{k-1}$, giving it a total sensitivity of $k+1$, a contradiction.

As $R(G, S_i) \geq \frac {3}{4}$ for $i \in [m-2,m]$,  we have that $\sum_{i=m-2}^m R(G, S_i) \geq \frac{9}{4}$. However, by the previous paragraph, $\frac{1}{2}$ of the vertices in $S_0$ have at most one neighbour in a heavy cube. Denote by $G'$ the vertices in $G$ neighbouring the other half of $S_0$. We now have that $\sum_{i=m-2}^m R(G', S_i) \geq \frac{7}{4}$. But then by the pigeonhole principle at least one vertex in $S_0$ would need to have 4 neighbours in $G$ in heavy cubes, which is impossible with only 3 heavy cubes.

\end{proof}

\end{document}